\documentclass[11pt]{article}
\usepackage{lmodern}
\usepackage{dsfont}
\usepackage{environ}
\input{definitions.sty}

\renewcommand{\Pr}{\mathop{\bf Pr\/}}

\def\realspos{\mathbb{R}_{\geq 0}}

\def\<{\langle}
\def\>{\rangle}

\def\vec{\bm}

\begin{document}

\title{A Framework for Single-Item NFT Auction Mechanism Design}

\author{
	    \textbf{Jason Milionis} \\
		\small Columbia University\\
		\small \tt{jm@cs.columbia.edu}
		\and
		\textbf{Dean Hirsch} \\
		\small Columbia University\\
		\small \tt{deanh@cs.columbia.edu}\\
		\and
		\textbf{Andy Arditi} \\
		\small Columbia University\\
		\small \tt{ava2123@columbia.edu}
		\and
		\textbf{Pranav Garimidi} \\
		\small Columbia University\\
		\small \tt{pg2682@columbia.edu}
}
\date{}
\maketitle

\thispagestyle{empty}

\begin{abstract}
Lately, Non-Fungible Tokens (NFTs), i.e., uniquely discernible assets on a blockchain, have skyrocketed in popularity by addressing a broad audience. However, the typical NFT auctioning procedures are conducted in various, ad hoc ways, while mostly ignoring the context that the blockchain provides, i.e., new possibilities, but at the same time new challenges in auction design. One of the main targets of this work is to shed light on the vastly unexplored design space of NFT Auction Mechanisms, especially in those characteristics that fundamentally differ from traditional and more contemporaneous forms of auctions. We focus on the case that bidders have a valuation for the auctioned NFT, i.e., what we term the \emph{single-item} NFT auction case. In this setting, we formally define an NFT Auction Mechanism, give the properties that we would ideally like a \emph{perfect} mechanism to satisfy (broadly known as incentive compatibility and collusion resistance) and prove that it is impossible to have such a \emph{perfect} mechanism. Even though we cannot have an all-powerful protocol like that, we move on to consider relaxed notions of those properties that we may desire the protocol to satisfy, as a trade-off between implementability and economic guarantees. Specifically, we define the notion of an \emph{equilibrium-truthful} auction, where neither the seller nor the bidders can improve their utility by acting non-truthfully, so long as the counter-party acts truthfully. We also define \emph{asymptotically second-price} auctions, in which the seller does not lose asymptotically any revenue in comparison to the theoretically-optimal (static) second-price sealed-bid auction, in the case that the bidders' valuations are drawn independently from some distribution. We showcase why these two are very desirable properties for an auction mechanism to enjoy, and construct the first known NFT Auction Mechanism which provably possesses such formal guarantees.
\end{abstract}

\blfootnote{Copyright \copyright\ 2022 J. Milionis, D. Hirsch, A. Arditi, P. Garimidi. Correspondence at: \href{mailto:jm@cs.columbia.edu}{jm@cs.columbia.edu}}

\section{Introduction}

Non-fungible tokens (or NFTs in shorthand notation) are blockchain assets (e.g., tokens) that are not interchangeable, i.e., that are ``incapable of mutual substitution'' \parencite{merriam-fungibility}.
Chaotic circumstances have repeatedly arisen in the past during public auctions of NFTs, and became especially prominent recently \parencite{recentNews}, whereby the creation of a disorganized market was blamed on ``the way (the organization) set up the sale and the inefficiency of its smart contract.'' This showcases the dire and long-overdue need for the study of the NFT Auction Mechanism Design Space to identify the unique aspects of such mechanisms and address the impending issues.

In this work, we provide a formal model for reasoning about how to perform an auction of the kind described above in a decentralized manner, through a blockchain. We focus on the case where bidders have a valuation of the NFT for which they are bidding, in which case the auction is labeled as a ``\emph{single-item NFT auction}.'' Note that this includes many common real-world auction formats involving \emph{predefined objects}, such as something that is in the physical world, e.g., a painting or an album, or a digital object, like the first tweet of Jack Dorsey \parencite{tweet_nft} or a digital art file \parencite{jpg_nft_beeple}.

Blockchains, in particular, are uniquely positioned to be able to offer brand new characteristics in such auctions. The single biggest advantage is self-evidently the decentralization: everyone is able to take part in any auction, and the auction's result is verifiable using on-chain data without any explicit trust assumptions to any third-party ``auctioneer.''
Additionally, there is the ability to be more creative in the types of rules allowed for bids (such as the ability to direct a portion of a bid someplace else other than the seller, as we will see below), hence broadening the design spectrum.
However, these innovative features do not come without a cost; in particular, blockchain implementations of auctions, besides allowing for a greater design space, allow for a far richer set of attack vectors that can be used to compromise or corrupt an auctioning procedure. For instance, due to the free participation in any auction and the pseudonymous identities, there may be ``\emph{fake bids},'' i.e., the seller of the auction may be incentivized to submit their own set of bids to the mechanism, if that could possibly lead to an increased profit for them.
Even more importantly, this leads to a great concern of directly applying the second-price auction mechanism to the decentralized setting we have been discussing: the seller can try to arbitrarily approximate the winning bid with their ``fake bids'' (that they will submit alongside the bids of actual bidders), thus forcing the winning bidder to pay their full bid instead of the optimal second highest \emph{actual} bid. This would essentially transform the auction to a first-price auction, which would no longer incentivize individual bidders to report their actual valuations for the item to the mechanism, but rather, to try and play ``guessing games'' for the bids that might be placed by others. This fact is well-known in the recent literature of transaction fee mechanism design, where a similar issue is present (see \Cref{sec:related-work} for more details).

In our new type of NFT auctioning framework, it is important to emphasize that the seller may now attempt to actively intermeddle in the auction according to their best interests; hence, the auction participants should include not only the bidders but also the seller.

Our goal in this framework is to answer the following questions:
\begin{enumerate}
\item How can we formally define a decentralized auctioning mechanism? What are the properties that would constitute \emph{perfect} theoretical security in such an auction?
\item Is it possible to create such a \emph{perfect} auction mechanism?
\item If this is not possible, then can the properties be relaxed such that they still provide meaningful economic security guarantees?
\item With these relaxed definitions in place, is it possible now to construct a viable auction mechanism, and if yes, what are the necessary core ideas of such a protocol?
\end{enumerate}

We address all these questions, and introduce our framework: consider that each bidder $i\in [n]$ (where by $[n]$ we denote the set $\{1,2,\dots,n\}$) has a (private) valuation $v_i \geq 0$ for the auctioned NFT, and reports (bids) $b_i \geq 0$ to the mechanism.

\begin{definition}
[Single-item NFT Auction Mechanism]
\label{def:single_mech}
A single-item NFT auction mechanism is described by a triplet $(x, \vec p, \vec r)$ where:
\begin{itemize}
    \item $x : \realspos^n \to [n]$ is the \textit{allocation rule} that determines which bidder receives an item, i.e., bidder $x(b_1, \dots, b_n) \in [n]$ is determined to be the single-item NFT auction winner, who gets the NFT.
    \item $\vec p : \realspos^n \to \realspos^n$ is the \textit{payment rule} which determines the complete payment that each bidder shall pay, i.e., bidder $i\in [n]$ pays amount $p_i(b_1, \dots, b_n) \leq b_i$ in total\footnote{Notice that this immediately means that submitting a bid of 0 causes a payment of 0.}.
    \item $\vec r : \realspos^n \to \realspos^n$ is the \textit{removal rule} (or else, burning rule, as it appears in the prior literature; as we note below, there is no need for the funds to be ``burnt'' in the standard way) specifying how much of the bidders' bids shall be removed from the current auctioning proceeds, i.e., an amount $r_i(b_1, \dots, b_n) \leq p_i(b_1, \dots, b_n)$ of the $i$-th bidder's payment $p_i(b_1, \dots, b_n)$ is ``removed'' from the possession of all auction participants (seller and bidders alike).
\end{itemize}
It is implied from the above definition that at the end of the auction, the seller receives an amount of $\sum\limits_{i=1}^n \left( p_i(b_1, \dots, b_n) - r_i(b_1, \dots, b_n) \right)$.
\end{definition}

To the best of our knowledge, \Cref{def:single_mech} is able to capture \emph{all} of the existing decentralized, single-item NFT auction designs that have been implemented or put forward to date.

Note that the removed amount is considered to be included in the payment rule, i.e., the payment rule subsumes all funds either removed from the bidder or given to the seller. We find this definition to be more natural, because the bidder does not directly care whether a spent amount from their payment is directed to the seller or not.

The crucial characteristic of the ``removal'' is that the funds need to be unavailable for the auction participants (the seller and the bidders); there is no need to burn them.
The utilities of the auction participants are
\begin{enumerate}
    \item For the seller, the utility is the amount that they receive:
    \[
    u_\text{seller}(b_1, \dots, b_n) = \sum\limits_{i=1}^n \left( p_i(b_1, \dots, b_n) - r_i(b_1, \dots, b_n) \right)
    \, .
    \]
    \item For the winning bidder $i = x(b_1, \dots, b_n)$, the utility is the difference between their true private valuation and the amount they paid: $u_i(b_1, \dots, b_n) = v_i - p_i(b_1, \dots, b_n)$.
    \item For the non-winning bidders $j \neq i \in [n]$, the utility is non-positive: $u_j(b_1, \dots, b_n) = -p_j(b_1, \dots, b_n)$. %
\end{enumerate}

We continue to intuitively review the three key properties that we would ideally like a \emph{perfect} auction mechanism to have.
These are: the seller incentive compatibility (the seller should not be able to submit \emph{fake bids} to the auction mechanism in order to improve his/her utility), the bidder incentive compatibility (traditionally known as strategy-proofness, it requires that the bidders have as a dominant strategy to bid truthfully, i.e., to set $b_i = v_i$), and the off-chain-agreement resistance (or OCA-proofness: that the bidders and the seller cannot improve their utility by colluding off-chain and agreeing to some arcane strategy of coordinating to bid strategically to the mechanism).
For the formal treatment of those, please refer to \Cref{sec:single_defs_props}.

However, we move on with a dire result: our \Cref{thm:noperfectauction} states that it is impossible to construct such a \emph{perfect} single-item auction mechanism, while still obtaining revenue from the auction proceeds. In other words, the only auction mechanisms that satisfy the above three properties are the \emph{trivial} ones, where no-one pays anything, and the seller receives no revenue. Clearly, this is an undesirable auction format, so \Cref{thm:noperfectauction} critically states that those commonly-thought desirable properties, even though they would able to defend against collusions of bidders and the seller and at the same time safeguard the interest of all involved parties, cannot be satisfied altogether; some kind of concession has to made in order to construct viable protocols.

In \Cref{sec:relaxed_oca} we discuss relaxed properties we can hope to achieve even when it is not possible to devise a perfectly collusion-proof protocol. We give the definition of an \emph{equilibrium-truthful} auction (\Cref{def:eq_truth}) and an \emph{asymptotically second-price} auction (\Cref{def:asym_sec_price}), for which we argue that, when combined, are good enough relaxations of desirable properties for an \emph{incentivizing, reasonable} NFT Auction Mechanism to possess.
\emph{Equilibrium-truthfulness} provides a formal truthfulness guarantee for the case of the auction having an untrustworthy seller without collusion among the bidders, and at the same time, if the seller is trustworthy, guarantees that the bidders will bid truthfully according to their true valuations. In other words, so long as the counter-party (bidders/seller) is legitimate, the other side of the auction participants (seller/bidders) will have as a dominant strategy to behave truthfully.
\emph{Asymptotic second-price} auction mechanisms give an \emph{approximation guarantee} for the revenue obtained through the auction: the seller will asymptotically enjoy the same revenue as the optimal static strategy-proof auction, which is the second-price sealed-bid auction.
In this way, any auction format that both satisfies equilibrium-truthfulness and is asymptotic second-price is attractive both for the seller (who will be able to obtain almost the best revenue he/she could hope for) and for the bidders (who will know that they do not need to collude with the seller, or fear that the seller will attempt to ``scam'' them into paying more than they should).

Perhaps most importantly, we show that we \emph{can} construct a protocol that is both equilibrium-truthful as well as asymptotically second-price. Our proposed protocol is described in detail in \Cref{sec:relaxed_oca} and its properties are provided in \Cref{theorem:protocol} and \textbf{formally proven} in \Cref{sec:analysis}. We believe that our proposed protocol is plausible to implement in current blockchains.

The core ideas behind our protocol are to:
\begin{enumerate}\label{enum:prot_main_ideas}
\item create a sealed-bid-like environment where bidders commit to their bids in advance without revealing their bid (by using some assumed-secure hiding commitment scheme),
\item incentivize the revelation of all commits to discourage the seller from trying to insert fake bids so as to attempt to convert the auction to a first-price one to maximize their revenue (by requiring an amount $L$ to be locked up in the contract along with any commit/bid, which will be returned to the bidder after revealing their bid) and properly punish the unrevealed bids (by completely \emph{removing} the locked amount $L$ from the possession of all auction participants: both the bidders and the seller), and
\item discourage the seller from selling the NFT to themselves, through possible very high fake bids to \emph{gain information} from this futile auction and attempt to perform another auction where he/she will have the information edge on the bidders (we economically prevent this by burning some of the money transferred from the winning bidder to the seller, according to a predefined \emph{fee function} $g(\text{second highest price})$).
\end{enumerate}

The definition of an asymptotic second-price auction (\Cref{def:asym_sec_price}) presupposes that the bidders' valuations are drawn independently from some distribution $D$.
Notice that, since the protocol we design will only depend on a reasonable lower bound $L$ of a quantity computable from the distribution $D$, the assumption that the bidders' private value stems from the same distribution $D$ is not particularly strong, because any reasonable heavy-tailed distribution with appropriate entropy and similar-to-expected characteristics can be utilized to compute the asked lower bound.

For the analysis of this protocol, we ignore potential gas fees; for example imagine the auction running on a highly-efficient L2 network. The impact of gas fees would have been a reduction of the utilities of the bidders, but if we further assume ``normality'' on the underlying chain's transaction fee mechanism, then the effect of those gas fees would be exactly the same across all bidders, resulting in a simple scaling of all the utilities.

\subsection{Existing NFT Auction Mechanisms}
An overwhelming majority of (single-item) NFT auctions currently take place off-chain, and hence fall short of the formal structure of our framework, as adumbrated above. These off-chain auctions generally take place on centralized web-based platforms, such as OpenSea. The entire auction takes place on this centralized platform, including the design choice of an auction format, as well as the bidding process. Once the auction is complete, the NFT is transferred to the winner, and the winner is charged accordingly. These final settlements are usually the only transactions to take place on-chain. A comparative advantage of off-chain auctions is that they minimize the number of on-chain transactions, and therefore minimize transaction fees. However, current off-chain auctions critically assume an almost complete level of trust in a third-party platform to execute the auction honestly (without almost any verifiability present), hence resembling traditional settings of auctions where there is a \emph{completely trustworthy auctioneer}. In short, currently-employed off-chain auctions crucially sacrifice security for convenience.

On the other hand, with an on-chain auction, the entire protocol is absolutely transparent and verifiable by the participants. Most importantly, bidders are not required to trust a third-party to run the auction. In this work, we specifically focus on the case when the seller might not be trusted, and so considering on-chain auctions is a natural idea. Currently, there are a few NFT marketplaces which offer on-chain auctions, including Foundation \parencite{foundation} and SuperRare \parencite{superrare}.

For both off-chain and on-chain auctions, the set of marketplaces which are popular today generally do not enjoy the truthfulness guarantee of actual second-price sealed-bid (Vickrey) auctions, something which would be very desirable to have. The most popular auction format for single-item NFTs appears to be an English auction format, where the current price is dynamically ascending according to the revealed bids of the users and the highest bid wins. In traditional environments this is paralleled to second-price auctions. However, in an environment of complete anonymity, such as the blockchain and the aforementioned marketplaces, where \emph{crucially} the seller may (and most likely will) be incentivized to maliciously affect the auction in order to enjoy maximum revenue, the English auction devolves into a first-price auction, because the seller may continually submit ``fake bids'' in order to artificially increase the price. The outcome will then approximate a first-price, non-truthful auction. Bidders do not have a clear idea how much they should bid. Further, when bidding on-chain, this can result in gas races where bidders rush to get their bids included \parencite{recentNews,vitalik}.

\subsection{Related Academic Work}
\label{sec:related-work}

From the auction mechanism design perspective with the unique characteristics of blockchains in mind, the first relevant setting \parencite{tim_eip_report,tim_tx_fee_mech_design_EC2021} is that of transaction fee mechanism design; there, we are interested in designing an auction mechanism for transaction fees (sometimes also called gas fees) on the blockchain.
There is, undoubtedly, some semblance in the single-item NFT Auction Mechanism Design case: in fact, our initial definitions and properties can be seen through this lens as smart-contract-mediated adaptations of those given by \textcite{tim_eip_report}, but with the crucial difference that the original definitions are on a setting where a single entity (the miner) has (mostly) dictatorial control over the allocation of winning bids.
NFT auctions critically differ on two aspects: first, on the auction-conducting smart contract's ability to \emph{completely specify} the outcome (the smart contract is akin to a ``third party'' that has the ability to regulate the distribution of NFTs; this introduces major mechanism feasibility differences with transaction fee mechanisms), and second,
on the availability and permanent recording of all bids or some appropriate (potentially private) transformation of them on the blockchain. The latter is also in stark contrast to the transaction fee mechanism design space, where \emph{crucially} not all gas fee bids can be recorded on the blockchain, but only the winning ones.
These differences make the respective design spaces significantly different from the perspective of realizable/feasible mechanisms, as we will move on to see from \Cref{sec:relaxed_oca} onwards.
The design space of transaction fee auctions is also examined in concurrent and independent work by \textcite{shiNov2021_tx_fee_mech_design}, where, in particular, the authors show that there is no transaction fee auction mechanism that satisfies all of the desiderata and always provides the miner with revenue from transactions.

For standard mechanism design, the problem of untrustworthy auctioneers has been studied by \textcite{credible}. They define the notion of a \emph{credible mechanism} where the auctioneer has a dominant strategy to follow the mechanism. For example, a Vickrey auction is not credible because the auctioneer has an incentive to fabricate a bid right below the highest bid to drive up their revenue. \textcite{credible} shows that, while a mechanism is desired to be sealed-bid, credible, and strategy-proof, only two out of the three properties can be satisfied if only winners of the auction make payments.
A natural question to ask is whether this impossibility result can be bypassed using cryptographic commitments to make it harder for the auctioneer to act untrustworthy. For the case of Vickrey auctions this would make it so that the auctioneer does not have knowledge of the bids' values until they are publicly known. One implementation of such a sealed-bid auction is to have two phases: a commit phase, and a reveal phase. In the commit phase, all bidders output some cryptographic commit to their bid value: $commit(bid||nonce)$. In the reveal phase, bidders reveal the pair $(bid,nonce)$. One issue with this implementation is that not all bids may be revealed in the reveal phase.
\textcite{ferreira} propose a solution to this problem by fining bidders for unrevealed commits, and paying these fines to the winning bidder. Our solution similarly intends to incentivize the revelation of all bids. A core difference between our proposal and that of \textcite{ferreira} is that we operate over a blockchain, inheriting its corresponding capabilities and assumptions; in particular, the anonymity of the seller as an entity (that allows them to provide \emph{fake bids} to the auction) means that providing the penalties to him/her would no longer guarantee their Incentive Compatibility. Even if one fixes this issue, the additional attack of the seller bidding to win the NFT themselves (just to gain information about the auction participants' bids) and repeatedly running auctions arises. We show a protocol that deals with these attacks, and formally prove that it is \emph{properly incentivizing all auction participants to behave truthfully}.

Lastly, \textcite{chen2022absnft} treat NFTs as asset-backed securities, allowing for fractional ownership of NFTs through what they call ``repurchase protocol'' which then enables performing repeated auctions of the same NFT. However, the proposed method is still vulnerable to common shortcomings of traditional auction mechanisms when applied to blockchain settings, as mentioned before.

\section{Properties and Impossibility Result}
\label{sec:single_defs_props}

Before we define the properties, we will need to formally describe an Off-Chain Agreement (OCA) for our purposes, along with the joint utilities of its participants. In particular, an OCA is a coalition of bidders $S\subseteq [n]$ with the seller $s$ who collude and submit a specific agreed-upon set of bids $\vec b' \in\realspos^k$ for some $k\ge 0$ (usually different from the ones that they would normally submit) to the auction mechanism such that in total, they would be better off in their joint utility (which is the sum of their individual utilities)\footnote{We assume that the OCA cannot consist of fully bypassing the auction mechanism, i.e., the NFT still has to be distributed through the mechanism, even, e.g., when all the parties agree to collude and only the agreed winner submits a fake bid of 0.}. In this way, they will be able to also obtain better individual utilities, for instance by splitting the joint utility differential that arises among all of them. More specifically, we define the joint utility of such an OCA as the aforementioned sum, where $\vec b_{-S}$ are the bids of the rest of the bidders that are not part of the coalition:
\[
u_{\text{joint}(S,s)} (\vec b', \vec b_{-S}) =
u_\text{seller}(\vec b', \vec b_{-S})
+
\sum_{i\in S} u_i(\vec b', \vec b_{-S})
\, .
\]
We shall write $u_{\text{joint}(i,s)}$ whenever we want to signify the particular coalition of the bidder $i$ with the seller and the meaning is clear from the context. We now move forward to provide the properties that one could possibly desire such an NFT auction mechanism to satisfy.

\begin{definition}
[Desirable properties of a single-item NFT Auction Mechanism]
\label{def:single_props}
We give the three properties
as follows:
\begin{enumerate}
\item \textit{Seller incentive-compatibility} (seller IC):
For any set of bids by bidders $b_1, \dots, b_n$,
if the seller submitted any set of ``fake bids'' $b_{n+1}, \dots, b_m$, then they would not be able to obtain any more utility than they already get without submitting any fake bids:
\begin{align*}
u_\text{seller}(b_1, \dots, b_n)
&=
\sum\limits_{i=1}^n \left( p_i(b_1, \dots, b_n) - r_i(b_1, \dots, b_n) \right)
\\ &\geq
\sum\limits_{i=1}^n \left( p_i(b_1, \dots, b_m) - r_i(b_1, \dots, b_m) \right)
\\ &-
\sum\limits_{j=n+1}^m r_j(b_1, \dots, b_m)
\, .
\end{align*}
\item \textit{Bidder incentive-compatibility} (bidder IC):
For any bidder $i\in [n]$, for any set of bids by (all but the $i$-th) bidders $b_1, \dots, b_{i-1}, b_{i+1} \dots, b_n$ (jointly denoted as $\vec b_{-i} \in \realspos^{n-1}$),
the $i$-th bidder's utility is maximized exactly when they bid their true valuation, i.e., for any potential bid $b_i$, it holds that
\[
u_i(b_1, \dots, b_{i-1}, v_i, b_{i+1} \dots, b_n)
\geq
u_i(b_1, \dots, b_{i-1}, b_i, b_{i+1} \dots, b_n)
\, .
\]
\item \textit{Off-Chain Agreement resistance} (OCA-proofness):
For any set of bidders $S \subseteq [n]$, bids by bidders not belonging to $S$: $b_{|S|+1}, \dots, b_n$ (jointly denoted as $\vec b_{-S} \in \realspos^{n-|S|}$),
if we consider the OCA of the seller $s$ with a set $S$ of bidders, then this coalition of auction participants is not able to obtain higher joint utility through any bids they might agree to submit, i.e., for any $k\geq 0$\footnote{The quantification ``any $k\ge 0$'' means that some of the colluding bidders may agree to ``disappear'' from the official auction mechanism, and only collude off-chain with their ``partners.'' If they \textit{all} agreed to \textit{submit} coordinated bids, then it would be $k=|S|$.} and for any $\vec b' \in \realspos^k$,
\[
u_{\text{joint}(S,s)} (\vec b', \vec b_{-S})
\leq
u_{\text{joint}(S,s)} (\vec b_S, \vec b_{-S})
\, ,
\]
where $\vec b_S \in \realspos^{|S|}$ are the bids that the colluding bidders would submit on their own without being part of that OCA.
\end{enumerate}
\end{definition}

We are now ready to state our impossibility result.

\begin{theorem}
[No perfect NFT auction mechanism]
\label{thm:noperfectauction}
There is no single-item NFT auction mechanism as considered in \Cref{def:single_mech}, satisfying the properties of bidder IC and OCA-proofness of \Cref{def:single_props}, where some bidder pays the seller some non-zero (positive) amount.
\end{theorem}

\Cref{thm:noperfectauction} (see proof in \Cref{app:proof_noperfectauction}) implies that the only feasible perfect mechanisms necessitate paying nothing to the seller, which is clearly undesirable as an auction format.

\section{Relaxation of properties and our proposed protocol}
\label{sec:relaxed_oca}

\Cref{thm:noperfectauction} implies that we cannot hope for the design of a \emph{perfect} auction mechanism. However, we may be interested to relax our provided guarantees such that the design of a mechanism with sufficiently secure ``relaxed'' properties is feasible. In particular,
we give a first relaxed definition that defends the individual incentives to bid truthfully given that there exists at least some side that is not rogue:
\begin{definition}
\label{def:eq_truth}
An auction is \emph{equilibrium-truthful} when:
\begin{enumerate}
    \item\label{def:eq_truth_seller} The seller's dominant strategy, assuming bidders bid truthfully (i.e. $b_i=v_i$), is to not post any fake bid, and
    \item\label{def:eq_truth_bidders} Each bidder's dominant strategy is to bid truthfully, assuming the seller is not posting any fake bids.
\end{enumerate}
\end{definition}

Further, in order for such an auction mechanism to be additionally considered \emph{reasonable}, it makes sense that we require that running such an auction will produce a favorable economic result for the seller. For this, we define the following auction property:

\begin{definition}
\label{def:asym_sec_price}
Given an auction in which at least $n$ participants are known to participate, and have their valuations independently drawn from the same distribution $D$, we say that an auction is \emph{asymptotically second-price} when the expected utility of the seller is $(1-o(1))\cdot \mathbb{E}[B_2]$, assuming bidders bid truthfully and that there are no fake bids, where $B_2$ is the second-highest bid, and $o(1)$ goes to 0 as $n\to\infty$.
\end{definition}

\begin{theorem}
\label{theorem:protocol}
Assume that there is a continuous distribution $D$ from which the valuations $v_i$ are independently drawn.
We denote by $f$ the probability density function of $D$, and by $F$ the cumulative distribution function.
Then, if
\begin{equation}
\label{eq:dist_req}
\sup_{s\in\operatorname{support}(f)} \frac{1-F(s)}{f(s)} < \infty
\,,
\end{equation}
then there exists a protocol that is both equilibrium-truthful and asymptotically second-price, as per \Cref{def:eq_truth,def:asym_sec_price}.
\end{theorem}

Notice that the requirement of \Cref{eq:dist_req} is satisfied for many natural distributions, including the uniform, exponential, and log-normal distributions.\footnote{More generally, it is satisfied for most unimodal distributions, such that f(s) = O(f'(s)).}
For the main ideas of the protocol, please refer to \Cref{enum:prot_main_ideas}. We now detail the protocol, and prove \Cref{theorem:protocol} in \Cref{sec:analysis}.

\begin{enumerate}
    \item The seller sets up a contract, to which bidders should send commits (i.e., $hash(bid||nonce)$), and the seller sends the NFT to the contract.
    \item Commit phase: Bidders send their commits to the contract, together with an amount $L$ (in tokens) to be locked up. Without both of these ingredients, the bid is not taken into account.
    \item After adequate time has passed, %
    e.g., enough blocks confirmed, the commit phase ends.
    \item Verification phase: Now bidders send to the contract their actual bid amounts (and nonce for verifying). We emphasize that the funds themselves are not sent to the contract at this point, but only the bid amounts the bidders committed to.
    \item At the end of the verification phase, for any unrevealed bid, their locked up amount $L$ is \emph{removed}.
    The contract can now look for the highest \emph{revealed} bid $b_1$ and second-highest \emph{revealed} bid $b_2$.
    \item An amount of $\min\{b_2, L\}$ from the winner remains at the contract for the seller.
    If $b_2\le L$, then the NFT is unlocked for transfer to the winner (the other bidders that revealed their commits will need to initiate a transaction to the contract to receive back their locked amounts). Otherwise, if $b_2 > L$, then an additional $b_2-L$ amount should be sent from the winner to the contract before the NFT is unlocked.
    The seller can withdraw $x-g(x)$, where $x$ is the funds received in the contract by the winner (including the initial $\min\{b_2, L\}$).
    \item After some predefined time, the auction ends. If the highest bidder did not transfer the funds, they still lose the $L<b_2$ they have locked up, and additionally do not receive the NFT.
\end{enumerate}

\section{Protocol Analysis} \label{sec:analysis}

The main purpose of this section is to give an overview of the proof of \Cref{theorem:protocol} (see \Cref{subsec:final_proof_of_protocol}).

\subsection{Fake Bids by the Seller}
We analyze the attack of multiple bids in the model without fees, to reason about the equilibrium truthfulness for the seller, i.e., \Cref{def:eq_truth}. Suppose that $B_1$ and $B_2$ are the highest and second-highest bids, respectively. Suppose the seller tries to submit multiple \emph{fake bids}. We will analyze under which conditions this attack does not benefit the seller's utility.
In particular,
the seller chooses a set $S$ of bids, and just before the verification phase ends, the seller gets to inspect all revealed bids and choose which of his/her fake bids in $S$ to reveal.

Observe that the seller does not lose anything by revealing all commits that are lower than $B_1$, and in fact will gain by getting to keep the locked up funds used for those bids.
Now, assume that the seller wants to sell the NFT in the current auction (the other case of repeated auction \emph{will be examined} in \Cref{subsec:repeated_auction})
they also must not reveal any commit higher than $B_1$. Therefore, the seller loses $L\cdot |\{s\in S:s>B_1\}|$, but gains $\max\{s\in S\cup\{B_2\}: s\le B_1\} - B_2$ (that is, the amount by which it raised the second-highest bid).

The seller therefore maximizes their utility as:
\[
\mathbb{E}\max\{0,\max\{s\in S: s\le B_1\} - B_2\} - L\cdot \mathbb{E}|\{s\in S:s>B_1\}|
\, ,
\]
where the expectation is over the values of $B_1$ and $B_2$. We note that by linearity of expectation, we can rewrite $\mathbb{E}|\{s\in S:s>B_1\}|$ as $\sum\limits_{s\in S}\Pr(B_1<s)$.
We have the following Lemma (see proof in \Cref{app:proof_singleton_fake_bid}):

\begin{lemma}
\label{lemma:singleton_fake_bid}
If there exists a set $S$ for the seller with a positive expected utility, then there also exists a set of size 1 with a positive expected utility.
\end{lemma}

Therefore, in order to ensure that the dominant strategy of the seller is not to place any fake bids, we need only ensure that there is no favorable \emph{single} fake bid, which happens when:

\begin{corollary}
The proposed protocol is \emph{not} susceptible to fake bids by the seller, as long as there is no repeated auction, if and only if
\[L\ge \max_{s}\frac{\mathbb{E}[s-B_2|B_2<s<B_1]\Pr(B_2<s<B_1)}{\Pr(s>B_1)}\,.\]
\end{corollary}

In the model where the private values of the $n$ bidders are drawn independently at random from a distribution $D$ with a probability density function $f(x)$ and cumulative distribution function $F(x)=\int_{-\infty}^x f(t)dt$, we have that
(we remind the reader that the probabilities are over the randomness of the $b_i$'s)
\[\Pr(B_1<s)=\Pr(\forall i: b_i < s) = F(s)^n\,.\]

Further, by symmetry,
\begin{align*}
\Pr(B_2<s<B_1)
&= n\cdot \Pr(B_2<s<B_1\wedge \text{bidder 1 is the winner})
\\ &=nF(s)^{n-1}(1-F(s))
\, .
\end{align*}

Also, using the fact that $\mathbb{E}[X]=\int_0^\infty \Pr(X>x)dx$ for any nonnegative random variable $X$, we get that
\begin{align*}
&\mathbb{E}(B_2|B_2<s<B_1)
=\int_0^s \Pr(B_2>x|B_2<s<B_1)dx
\\ &= \frac{1}{\Pr(B_2<s<B_1)}\int_0^s \Pr(x<B_2<s<B_1)dx
\\ &= \frac{1}{\Pr(B_2<s<B_1)}\int_0^s (\Pr(B_2<s<B_1) - \Pr(B_2<x<B_1))dx
\\ &= s - \frac{1}{nF(s)^{n-1}(1-F(s))}\int_0^s nF(x)^{n-1}(1-F(x)) dx
\\ &= s - \frac{1}{F(s)^{n-1}(1-F(s))}\int_0^s F(x)^{n-1}(1-F(x)) dx
\end{align*}

Substituting the formulae to the bound on $L$,
we obtain the simplified corollary:

\begin{corollary}
When all bidders draw their private value independently at random from a
distribution with a cumulative probability function $F$, the auction is equilibrium-truthful for the seller so long as there is no repeated auction exactly when
\[L\ge \max_{s} \frac{1}{F(s)^n}\int_0^s F(x)^{n-1}(1-F(x)) dx \, , \]
where the maximum is taken over all values of $s$ in the support of $f$.
\end{corollary}

Under a distribution with specified lower bound on its support, the following Theorem can be shown (see \Cref{app:proof_L_lower_bound}):

\begin{theorem} \label{theorem:L_lower_bound}
Suppose that all bidders draw their private value independently at random from a
distribution supported on $[a,\infty)$ for some $a>0$, with a cumulative probability function $F$ and density function $f$, such that $\frac{1-F(s)}{f(s)}$ is bounded from above on the support of $f$.
Let $L(n)$ be the lower bound on $L$ for $n$ bidders.
Then, the auction is equilibrium-truthful for the seller so long as there is no repeated auction, if
$L(n)\sim \frac{1}{n}\cdot \sup_{s} \frac{1-F(s)}{f(s)}$ as $n\to\infty$.
\end{theorem}

\subsection{Repeated Auction}
\label{subsec:repeated_auction}

Let, again, the highest bid be $B_1$ and the second-highest bid be $B_2$, unknown to the seller.
We will analyze the conditions that make the seller's utility be maximized without any repeated auction.

In the first auction, the seller pays $B_1$ and receives $B_1-g(B_1)$ (where $g$ is the fee function), so in total they lose $g(B_1)$. Had the seller not done this attack, the total gain would have been $B_2-g(B_2)$.

In the second auction, the seller posts $B_1-\varepsilon$ as a bid, thus receiving essentially $B_1-g(B_1)$ for the item. Thus, in total, instead of receiving $B_2-g(B_2)$, the seller receives $B_1-2g(B_1)$. So the seller is acting truthfully if and only if
\[ B_2-g(B_2) \ge B_1-2g(B_1) \Leftrightarrow B_1-B_2 \le 2g(B_1)-g(B_2) \,. \]

Thus, if we find a fee function $g$ such that $B_1-B_2 \le 2g(B_1)-g(B_2)$ (in expectation), the seller would be acting truthfully, and not try to insert fake bids so as to perform repeated auctions. Assuming a nondecreasing fee function, where $g(B_1) \ge g(B_2)$, it is enough to require that $B_1-B_2\le g(B_1)$, i.e., that the fee provides an estimate for an upper bound on the gap between the two highest bids.

\begin{theorem}
[\Cref{app:proof_fee_asymptotics}]
\label{theorem:fee}
If the distribution $D$ satisfies
\[\sup_{x\in\operatorname{support}(f)} \frac{1-F(x)}{f(x)} < \infty\, ,\]
then choosing the fee function $g(x)=\alpha x$ for an appropriate $\alpha = o(1)$, i.e., such that $\lim\limits_{n\to\infty}\alpha(n)=0$,
guarantees that the seller's dominant strategy is to not run repeated auctions by inserting fake bids.
\end{theorem}

\subsection{Final Proof of \Cref{theorem:protocol}}
\label{subsec:final_proof_of_protocol}

\begin{proof}
We prove that the protocol described with choosing the locked-amount $L$ as in \Cref{theorem:L_lower_bound} and the fee function $g$ as in \Cref{theorem:fee}, satisfies the theorem.

We first show that the auction described is equillibrium-truthful.
Indeed, for \Cref{def:eq_truth_bidders} of \Cref{def:eq_truth}, suppose that the seller does not place any fake bids. Then, since the amount each bidder pays in case they win is the second highest price, the auction satisfies the conditions for the bidder IC of second-price auction, and it follows that the best strategy of any bidder is to bid truthfully (i.e., $b_i=v_i$).
On the other hand, for \Cref{def:eq_truth_seller} of \Cref{def:eq_truth}, assuming all bids are truthful, the seller's dominant strategy is to not place any fake bids, as shown in \Cref{theorem:L_lower_bound,theorem:fee}.

Finally, \Cref{theorem:fee} additionally shows that the auction is asymptotically second-price, since the utility for the seller is $(1-\alpha)\mathbb{E}[B_2]$ with $\alpha=o(1)$.
\end{proof}

\printbibliography

\appendix

\section{Proof of Theorem~\ref{thm:noperfectauction}}
\label{app:proof_noperfectauction}

The backbone of the proof is based on the interesting observation which we will prove below in \Cref{lemma:perfect_mech} that any mechanism according to \Cref{def:single_mech} satisfying bidder IC and OCA-proofness (from \Cref{def:single_props}) must be such as to remove all funds, i.e., $\vec r = \vec p$.
Then, it is evident that no desirable mechanism can exist, because the dependence of the removed amount on the current bids is detrimental to the OCA-proofness of the mechanism: in particular, suppose that such a mechanism exists.
Now, consider the OCA of all bidders with the seller: all participants of the OCA agree that only the normal winner (say $i\in [n]$) of the auction will submit a single bid of $b=0$ (and, by the assumptions of \Cref{def:single_mech}, must pay 0) to the smart contract acting as the auctioneer, and that winner will pay the same amount that a typical auction with truthful valuations would want them to pay, i.e., $p_i(v_1, \dots, v_n)$ where $v_j\ \forall j\in [n]$ are considered to be the true private valuations of the bidders, but instead now this payment will stay within the coalition
(and may be distributed to the members of the OCA in such a way that every member, who would previously obtain zero, has an incentive to participate).
Thus, the joint utility of the OCA will be $(v_i - p_i(v_1, \dots, v_n)) + p_i(v_1, \dots, v_n) = v_i$
whereas their utility under the auction mechanism would have been $v_i - \sum\limits_{j=1}^n r_j(v_1, \dots, v_n)$. Thus, if even one bidder $j$ has amount $r_j(b_1, \dots, b_n) > 0$ removed, then the auction mechanism would not be OCA-proof. By \Cref{lemma:perfect_mech}, this also means that no bidder must pay any non-zero amount.
However, by the premises of \Cref{thm:noperfectauction}, it must be the case that some bidder pays some $p_i(b_1, \dots, b_n) > 0$. Hence, this is a contradiction, and there is no such mechanism.

\begin{lemma}
\label{lemma:perfect_mech}
Any mechanism according to \Cref{def:single_mech} satisfying bidder IC and OCA-proofness (from \Cref{def:single_props}) must be such as to remove all funds, i.e., $\vec r = \vec p$.
\end{lemma}
\begin{proof}
We now move on to the second part of the proof, i.e., proving that if the mechanism is required to not match exactly the full removal rule given above ($\vec r = \vec p$), then no possible mechanism can simultaneously satisfy the properties of bidder IC and OCA-proofness. First, we make the critical observation that OCA-proofness implies the classical Individual Rationality property from traditional auction theory (non-winning bidders pay zero)\footnote{Alternatively to the proof by OCA, this could also arise just from bidder IC: if a bidder bids 0, then they must pay 0 by \Cref{def:single_mech}, thus by bidder IC, the utility of every bidder who bids truthfully must be non-negative, which means that non-winning, truthful bidders must pay 0.}, by contradiction: suppose that there was a bidder $i\in [n]$ that bid $b_i$ such that $x(b_1, \dots, b_n) \neq i$ (they were not the auction winner), then the OCA of that bidder (along with any other non-winning bidder for whom the Individual Rationality property does not hold, i.e., they pay a strictly positive amount) with the winning bidder and the seller would obtain strictly greater joint utility by not submitting (or equivalently, submitting zero) bids for the users whose payments would not satisfy the Individual Rationality property, since the joint utility will be greater by that (positive) amount. Hence, a mechanism that is not Individually Rational would necessarily not be OCA-proof, which is a contradiction. Hence, the conclusion is that the mechanism in question has to make only the winning bidder pay a potentially non-zero amount (and potentially remove some portion of that amount). Individual Rationality further implies that the only source of revenue for the seller is the payment (minus the removed amount) of the winning bidder.

To proceed, we now remark that the payment rule $\vec p$ alongside the allocation rule $x$ are subject to Myerson's Lemma \parencite{myerson} since they have to satisfy the property of bidder IC, which is precisely the sense of DSIC (dominant strategy incentive compatibility) used in auction theory. In particular, this has two implications: (we will refer to them as property 1 and 2, respectively, in what follows)
\begin{enumerate}
    \item The allocation rule $x$ is monotone, which in our notation, means that if for some bids $b_1, \dots, b_n$ and some $i \in [n]$, it holds that $x(b_1, \dots, b_n) = i$, then $\forall b_i' \geq b_i : x(b_i', \vec b_{-i}) = i$ and also if for some bids $b_1, \dots, b_n$ and some $i \in [n]$, it holds that $x(b_1, \dots, b_n) \neq i$, then $\forall b_i' \leq b_i : x(b_i', \vec b_{-i}) \neq i$.
    \item The payment rule $p_i(b_1, \dots, b_n)$ where $i = x(b_1, \dots, b_n)$ is the unique rule that imposes the payment which is the minimum bid $b$ such that $x(b,\vec b_{-i}) = i$.\footnote{This is because a bid of 0 is guaranteed to pay 0.}
\end{enumerate}

The rest of the proof will proceed on the basis of a contradiction: assume that there is a mechanism such that $u_\text{seller}(b_1, \dots, b_n) > 0$ for some \textit{specific} bids $b_1, \dots, b_n$ (this is equivalent to saying that there are bids that make the supposed equality $\vec r = \vec p$ not hold) and prove that this is inconsistent with the properties of OCA-proofness and the above 2 properties arising from the bidder IC. In particular, the target is, of course, to construct an OCA that would violate the OCA-proofness of the mechanism. We shall construct the simplest possible OCA: that of one bidder with the seller. Target: without the OCA, the bidder would be hopeless (cannot get the item), but with the OCA, the bidder can now magically get the item, and the coalition will marginally benefit from the OCA (i.e., the seller's utility will outbalance any loss from the individual utility of the bidder, since they would not get the item in the first place without the OCA). We will construct such an OCA. According to that, we analyze the joint utilities and reverse engineer:

\begin{itemize}
    \item With the OCA's false, agreed-upon bid $b_i'$ (everybody else bids their true values, due to bidder IC): $u_{\text{joint}(i, s)}(b_i', \vec v_{-i}) = u_\text{seller}(b_i', \vec v_{-i}) + (v_i - p_i(b_i', \vec v_{-i}))$.
    \item Without the OCA, due to bidder IC, the bidder would bid their true value, insufficient to get them the item (but at the same time, they pay zero as we proved above, thus obtaining zero utility): $u_{\text{joint}(i, s)}(v_i, \vec v_{-i}) = u_\text{seller}(v_i, \vec v_{-i}) + (0)$.
\end{itemize}

We want to construct the OCA, i.e., make it so that \[u_\text{seller}(b_i', \vec v_{-i}) + (v_i - p_i(b_i', \vec v_{-i})) > u_\text{seller}(v_i, \vec v_{-i})\,.\] If $b_i' < v_i$ ($i$ agreed to underbid), then $i$ would not get the object (due to the $x$'s monotonicity as above property 1), so $i$ must overbid, and thus on its own, would obtain negative utility. Thus, we set its valuation to be $v_i = p_i(b_i', v_{-i}) - \alpha$ for some $\alpha > 0$ to be discovered later, and we now only want $\alpha < u_\text{seller}(b_i', \vec v_{-i}) - u_\text{seller}(v_i, \vec v_{-i})$ for the OCA to game the mechanism.

By the assumption of the \textit{specific} bids that make the removed amount not be equal to the whole payment, we have that there exist $b_1, \dots, b_n$ such that $u_\text{seller}(b_1, \dots, b_n) > 0$. Then, we claim that there exist \textit{specific} bids $b_1', b_2', \dots, b_n'$ and bidder $i\in [n]$ such that $u_\text{seller}(b_i', \vec b_{-i}') > u_\text{seller}(0, \vec b_{i-}')$. Proof by contradiction: if there were not, then for all bids $b_i'$ and $i\in [n]$, it would be true that $u_\text{seller}(b_i', \vec 0) \leq 0 \Rightarrow u_\text{seller}(b_i', \vec 0) = 0$ and since this is for all $i$, $u_\text{seller}(b_1, \dots, b_n) = 0$ for all bids $b_1, \dots, b_n$ which is a contradiction since we assumed there are some bids that do not satisfy this.

We will now use these specific values $b_1', b_2', \dots, b_n'$ that we proved they exist: define the valuations of other users $\vec v_{-i} = \vec b_{-i}'$, and $i$'s as has already been discussed ($\alpha$ is still free to be chosen later). The final observation needed is that the specific bid of a bidder should not matter for the determination of the removal rule $\vec r$ as long as they are ``in the same range of item-getting'' (i.e., if both a current bid and an alternative bid would either make them win or lose). Formally, we will prove through this observation that $u_\text{seller}(v_i, \vec v_{-i}) = u_\text{seller}(0, \vec v_{-i})$ (because this also holds for the payment rule, by properties 1 and 2 above), i.e., it does not matter to the seller whether a non-winning bidder bid some $v_i > 0$ or exactly zero. Then, set any $\alpha$ such that $0 < \alpha < u_\text{seller}(b_i', \vec b_{-i}') - u_\text{seller}(0, \vec b_{-i}')$ and the OCA is complete, thus our proof by contradiction is completed.

We now proceed to prove the above observation:
we have that if $x(b_1, \dots, b_n) \neq i$ and $x(b_i', \vec b_{-i}) \neq i$, then by \Cref{def:single_mech} and OCA-proofness as above, $r_i(b_1, \dots, b_n) = r_i(b_i', \vec b_{-i}) = 0$. In the other case that $x(b_1, \dots, b_n) = x(b_i', \vec b_{-i}) = i$, if we assume that $r(b_i, \vec b_{-i}) > r(b_i', \vec b_{-i})$ then the mechanism would not be OCA-proof: consider the situation where all bidders have valuations $v_j = b_j$ for all $j\in [n]$ and the OCA of $i$ with the seller: $u_{\text{joint}(i,s)}(b_i, \vec b_{-i}) = u_\text{seller}(b_i, \vec b_{-i}) + (b_i - p_i(b_i, \vec b_{-i})) = b_i - r(b_i, \vec b_{-i})$ but could be manipulated with the OCA because the payments remain the same:
\begin{align*}
u_{\text{joint}(i,s)}(b_i', \vec b_{-i})
&= u_\text{seller}(b_i', \vec b_{-i}) + (b_i - p_i(b_i', \vec b_{-i}))
\\ &= b_i - r(b_i', \vec b_{-i})
\\ &> b_i - r(b_i, \vec b_{-i}) 
\\ &= u_{\text{joint}(i,s)}(b_i, \vec b_{-i})
\, .
\end{align*}
This is a contradiction to the OCA-proofness of the mechanism. (the other case that $r(b_i, \vec b_{-i}) < r(b_i', \vec b_{-i})$ also leads to the same contradiction by taking $v_i = b_i'$, but $v_j = b_j$ for $j\neq i$)
\end{proof}

\section{Proof of Lemma~\ref{lemma:singleton_fake_bid}}
\label{app:proof_singleton_fake_bid}

Suppose there is a set with a positive expected utility. Then there is a set with minimal size among all such sets. This set is nonempty, since for the empty set the expected utility is 0. Let this set be $S$. Assume for the sake of contradiction that $|S|>1$, and take $s_1=\max S$ and let $s_2=\max\{s\in S: s < s_1\}$.

By the assumption that $S$ has \emph{minimal size} among all sets with positive expected utility, if $s_1$ is taken out, the utility becomes non-positive, and in particular decreases. On the other hand, the expected utility increases by $L\Pr(s_1>B_1)$ and decreases only when $B_2<s_1<B_1$ by: $\mathbb{E}[s_1-B_2|s_2<B_2<s_1<B_1]\Pr(s_2<B_2<s_1<B_1) + (s_1-s_2)\Pr(B_2<s_2<s_1<B_1)$ (where we used the law of total expectation).

The claim that the set $\{s_1\}$ has a positive utility is equivalent to the claim that $L\Pr(s_1>B_1)<\mathbb{E}[s_1-B_2|B_2<s_1<B_1]\Pr(B_2<s_1<B_1)$, so the Lemma will follow if we show that
\[\mathbb{E}[s_1-B_2|B_2<s_1<B_1]\Pr(B_2<s_1<B_1)\]
is at least
\begin{align*}
    &\mathbb{E}[s_1-B_2|s_2<B_2<s_1<B_1]\Pr(s_2<B_2<s_1<B_1)\\ &+ (s_1-s_2)\Pr(B_2<s_2<s_1<B_1)
\end{align*}

Indeed, by further conditioning on $B_2$ and using the law of total expectation, we can see that the first value equals
\begin{align*}&\mathbb{E}[s_1-B_2|s_2<B_2<s_1<B_1]\Pr(s_2<B_2<s_1<B_1) \\
&+ (s_1-B_2)\Pr(B_2<s_2<s_1<B_1)
\, ,
\end{align*}
which is larger than the second value because $s_1-B_2 > s_1-s_2$.

\section{Proof of Theorem~\ref{theorem:L_lower_bound}}
\label{app:proof_L_lower_bound}

We first fix a value of $s$.
By Fubini, we can switch between the limit of $n\to \infty$ and the supremum over $s$, if by first finding the limit for fixed $s$ and then taking the maximum over the support, we can obtain a finite result.

For a fixed value of $s$, we make a change of variables $t=F(x)$ in the integral, so $dx=\frac{dt}{f(F^{-1}(t))}$. Thus the integral is:
\[\int_0^s F(x)^{n-1}(1-F(x))dx = \int_0^{F(s)} t^{n-1}(1-t)\frac{dt}{f(F^{-1}(t))}\,.\]

We now note that the part of the integral that is bounded above by $F(s)-\varepsilon$ is thus also upper-bounded by $O((F(s)-\varepsilon)^n)$ and would therefore not contribute asymptotically even after dividing by $F(s)^n$ (note that we have $t<1$ in the integral). Thus we can approximate this integral asymptotically by replacing $(1-t)/f(F^{-1}(t))$ with its value on the upper limit of the integral; that is, replace $(1-t)\frac{dt}{f(F^{-1}(t))}$ with $(1-F(s))\frac{dt}{f(F^{-1}(F(s)))} = (1-F(s))\frac{dt}{f(s)}$. Now the integral becomes
\begin{align*}
    \int_0^s F(x)^{n-1}(1-F(x))dx &\approx \frac{1-F(s)}{f(s)}\int_0^{F(s)} t^{n-1}dt \\ &= \frac{1-F(s)}{nf(s)}F(s)^n
    \,.
\end{align*}

Thus, dividing by $F(s)^n$ we finally obtain the limit as $L(n,s)=\frac{1-F(s)}{nf(s)}$. Taking the supremum over $s$ yields the stated result.

\section{Proof of Theorem~\ref{theorem:fee}}
\label{app:proof_fee_asymptotics}

We find $\mathbb{E}[B_1]-\mathbb{E}[B_2]$. For that, we use the CDF of $B_1$ $F_{B_1}(x)=F(x)^{n}$, and the CDF of $B_2$ $F_{B_2}(x)=\Pr(B_1<x) + \Pr(B_1>x\wedge B_2<x)=F(x)^n + nF(x)^{n-1}(1-F(x))$. Thus, since for a nonnegative random variable $X$ we have $\mathbb{E}[X]=\int_0^\infty \Pr(X\ge x)dx=\int_0^\infty (1-F(x))dx$, we get that
\begin{align*}
&\mathbb{E}[B_1]-\mathbb{E}[B_2] \\ &= \int_0^\infty \left(1-F(x)^n - (1-F(x)^n-nF(x)^{n-1}(1-F(x))\right) dx\\
&= n\int_0^\infty F(x)^{n-1}(1-F(x)) dx\\ &= n\int_0^\infty F(x)^{n-1}\frac{1-F(x)}{f(x)} f(x) dx
\end{align*}

Let $A=\sup_x \frac{1-F(x)}{f(x)}$, which is finite by the Theorem's premise. Then we finally obtain
\[\mathbb{E}[B_1]-\mathbb{E}[B_2] \le nA\int_0^\infty F(x)^{n-1}f(x)dx = A F(x)^n\Big|^\infty_0=A\]

We will now show that we can always choose
\[\alpha = 2 \frac{\mathbb{E}[B_1-B_2]}{\mathbb{E}[B_2]}\]

We evidently have $\mathbb{E}[B_1-B_2] < \mathbb{E}[g(B_2)]$.

To show that $\lim\limits_{n\to\infty} \alpha(n) = 0$, we distinguish between two cases:
\begin{itemize}
    \item $\lim\limits_{n\to\infty} \mathbb{E}[B_1] = \infty$.
    In this case, we also have $\lim\limits_{n\to\infty} \mathbb{E}[B_2] = \infty$, since we have shown that the difference $\mathbb{E}[B_1]-\mathbb{E}[B_2]$ is finite.
    Therefore $\alpha(n) \le 2 \frac{A}{\mathbb{E}[B_2]} \to 0$.
    
    \item $\lim\limits_{n\to\infty} \mathbb{E}[B_1] = M < \infty$ for some finite $M$. We note that the limit necessarily exists, since this expectation is nondecreasing with $n$. In this case, $\mathbb{E}[B_2]$ also has the same limit $M$ (as shown below), hence the numerator approaches 0 and the denominator is nondecreasing, giving $\lim\limits_{n\to\infty} \alpha(n) = 0$.
    
    It remains to prove that $\lim\limits_{n\to\infty} \mathbb{E}[B_2] = M$. Since $B_2\le B_1$, it suffices to prove that $\liminf \mathbb{E}[B_2] \ge M$. Indeed, we have that $B_2(n)$ is at least the minimum between the best bids between any two arbitrary disjoint sets of the $n$ bids. In particular, we obtain that $B_2(n)$ is at least the minimum between the maximum of the first $n/2$ bids and the maximum of the last $n/2$ bids.
    Thus if we show that for all $\varepsilon>0$ we have that $\lim_{n\to\infty}\Pr(B_1 < M-\varepsilon) = 0$, it would follow that $\liminf \mathbb{E}[B_2(n)] \ge M-\varepsilon$, for any $\varepsilon>0$, so $\liminf \mathbb{E}[B_2(n)] \ge M$ and we will be done.
    
    For that, we notice that $\Pr(B_1<M-\varepsilon) = F(M-\varepsilon)^n$, so it remains to show that $F(M-\varepsilon)<1$. Indeed, if $F(M-\varepsilon)=1$, we would also have $\Pr(B_1<M-\varepsilon)=1$ exactly (i.e. not in the limit), hence $\mathbb{E}[B_1]\le M-\varepsilon$, a contradiction.

\end{itemize}

This concludes the proof.

\end{document}